\documentclass[conference]{IEEEtran}
\usepackage{amsmath,amssymb,amsthm,cite,graphicx,array}
\usepackage{graphicx}
\usepackage{float}
\usepackage{caption}
\usepackage{multicol}
\usepackage{mathrsfs}
\usepackage{amsmath}
\usepackage{amssymb}
\usepackage{amsthm}
\usepackage{multicol}
\usepackage{algorithm}
\usepackage[noend]{algpseudocode}
\usepackage{float}
\usepackage{placeins}
\usepackage{epstopdf}
\usepackage{multirow}
\usepackage{subfigure}
\usepackage{enumitem}
\usepackage[utf8]{inputenc}
\floatstyle{ruled}
\newfloat{algorithm}{tbp}{loa}
\providecommand{\algorithmname}{Algorithm}
\floatname{algorithm}{\protect\algorithmname}
\theoremstyle{remark}
\newtheorem{theorem}{Theorem}
\newtheorem{corollary}{Corollary}
\newtheorem{lemma}{Lemma}
\newtheorem{definition}{Definition}
\newtheorem{remark}{Remark}
\theoremstyle{remark}
\newtheorem{example}{Example}
\ifCLASSINFOpdf
\else
\fi

\title{A New Upperbound on the Broadcast Rate of Index Coding Problems with Symmetric Neighboring Interference}

\begin{document}

\author{Mahesh~Babu~Vaddi~and~B.~Sundar~Rajan\\ 
 Department of Electrical Communication Engineering, Indian Institute of Science, Bengaluru 560012, KA, India \\ E-mail:~\{mahesh,~bsrajan\}@ece.iisc.ernet.in }
 
\maketitle
\begin{abstract}
A single unicast index coding problem (SUICP) with symmetric neighboring interference (SNI) has equal number of $K$ messages and $K$ receivers, the $k$th receiver $R_{k}$ wanting the $k$th message $x_{k}$ and having the side-information $\mathcal{K}_{k}=(\mathcal{I}_{k} \cup x_{k})^c,$ where ${I}_k= \{x_{k-U},\dots,x_{k-2},x_{k-1}\}\cup\{x_{k+1}, x_{k+2},\dots,x_{k+D}\}$ is the interference with $D$ messages after and $U$ messages before its desired message. The single unicast index coding problem with symmetric neighboring interference (SUICP-SNI) is motivated by topological interference management problems in wireless communication networks. Maleki, Cadambe and Jafar obtained the capacity of this SUICP-SNI with $K$ tending to infinity and Blasiak, Kleinberg and Lubetzky for the special case of $(D=U=1)$ with $K$ being finite. Finding the capacity of the SUICP-SNI for arbitrary $K,D$ and $U$ is a challenging open problem. In our previous work, for an SUICP-SNI with arbitrary $K,D$ and $U$, we defined a set $\mathcal{\mathbf{S}}$ of $2$-tuples such that for every $(a,b)$ in that set $\mathcal{\mathbf{S}}$, the rate $D+1+\frac{a}{b}$ is achieved by using vector linear index codes over every finite field. In this paper, we give an algorithm to find the values of $a$ and $b$ such that $(a,b) \in \mathcal{\mathbf{S}}$ and $\frac{a}{b}$ is minimum. We present a new upperbound on the broadcast rate of SUICP-SNI and prove that this upper bound  coincides with the existing results on the exact value of the capacity of SUICP-SNI in the respective settings.
\end{abstract}
\section{Introduction and Background}
\label{sec1}
\IEEEPARstart {A}{n} index coding problem, comprises a transmitter that has a set of $K$ independent messages, $X=\{ x_0,x_1,\ldots,x_{K-1}\}$, and a set of $M$ receivers, $R=\{ R_0,R_1,\ldots,R_{M-1}\}$. Each receiver, $R_k=(\mathcal{K}_k,\mathcal{W}_k)$, knows a subset of messages, $\mathcal{K}_k \subset X$, called its \textit{Known-set} or the \textit{side-information}, and demands to know another subset of messages, $\mathcal{W}_k \subseteq \mathcal{K}_k^\mathsf{c}$, called its \textit{Want-set} or \textit{Demand-set}. The transmitter can take cognizance of the side-information of the receivers and broadcast coded messages, called the index code, over a noiseless channel. The objective is to minimize the number of coded transmissions, called the length of the index code, such that each receiver can decode its demanded message using its side-information and the coded messages.

The problem of index coding with side-information was introduced by Birk and Kol \cite{ISCO}. Ong and Ho \cite{OnH} classified the binary index coding problem depending on the demands and the side-information possessed by the receivers. An index coding problem is unicast if the demand-sets of the receivers are disjoint. An index coding problem is single unicast if the demand-sets of the receivers are disjoint and the cardinality of demand-set of every receiver is one. Any unicast index coding problem can be converted into an equivalent single unicast index coding problem. A single unicast index coding problem (SUICP) can be described as follows: Let $\{x_{0}$,$x_{1}$,\ldots,$x_{K-1}\}$ be the $K$ messages, $\{R_{0}$,$R_{1},\ldots,R_{K-1}\}$ are $K$ receivers and $x_k \in \mathcal{A}$ for some alphabet $\mathcal{A}$ and $k=0,1,\ldots,K-1$. Receiver $R_{k}$ wants the message $x_{k}$ and knows a subset of messages in $\{x_{0}$,$x_{1}$,\ldots,$x_{K-1}\}$ as side-information. 


A solution (includes both linear and nonlinear) of the index coding problem must specify a finite alphabet $\mathcal{A}_P$ to be used by the transmitter, and an encoding scheme $\varepsilon:\mathcal{A}^{t} \rightarrow \mathcal{A}_{P}$ such that every receiver is able to decode the wanted message from $\varepsilon(x_0,x_1,\ldots,x_{K-1})$ and the known information. The minimum encoding length $l=\lceil log_{2}|\mathcal{A}_{P}|\rceil$ for messages that are $t$ bit long ($\vert\mathcal{A}\vert=2^t$) is denoted by $\beta_{t}(G)$. The broadcast rate of the index coding problem with side-information graph $G$ is defined \cite{ICVLP} as,
$\beta(G) \triangleq   \inf_{t} \frac{\beta_{t}(G)}{t}.$
If $t = 1$, it is called scalar broadcast rate. For a given index coding problem, the broadcast rate $\beta(G)$ is the minimum number of index code symbols per message symbol required to transmit to satisfy the demands of all the receivers. The capacity $C(G)$ for the index coding problem is defined as the maximum number of message symbols transmitted per index code symbol such that every receiver gets its wanted message symbols and all the receivers get equal number of wanted message symbols. The broadcast rate and capacity are related as 
\begin{center}	
$C(G)=\dfrac{1}{\beta(G)}$.
\end{center}  

Instead of one transmitter and $K$ receivers, the SUICP can also be viewed as $K$ source-receiver pairs with all $K$ sources connected with all $K$ receivers through a common finite capacity channel and all source-receiver pairs connected with either zero of infinite capacity channels. This problem is called multiple unicast index coding problem in \cite{MCJ}.
\subsection{Single unicast index coding problem with symmetric neighboring interference}
A single unicast index coding problem with symmetric neighboring interference (SUICP-SNI) with equal number of $K$ messages and receivers, is one with each receiver having a total of $U+D<K$ interference, corresponding to the $D~(U \leq D)$ messages above and $U$ messages before its desired message. In this setting, the $k$th receiver $R_{k}$ demands the message $x_{k}$ having the interference
\begin{equation}
\label{antidote}
{I}_k= \{x_{k-U},\dots,x_{k-2},x_{k-1}\}\cup\{x_{k+1}, x_{k+2},\dots,x_{k+D}\}, 
\end{equation}
\noindent
the side-information being 
$\mathcal{K}_{k}=(\mathcal{I}_{k} \cup x_{k})^c.$

Maleki \textit{et al.} \cite{MCJ} found the capacity of SUICP-SNI with $K\rightarrow \infty$ to be  
\begin{align}
\label{cap1}
C=\frac{1}{D+1},
\end{align}
and an upper bound for the capacity of SUICP-SNI for finite $K$ to be 
\begin{align}
\label{outerbound}
C \leq \frac{1}{D+1},
\end{align}
which is same as the broadcast rate of the SUICP-SNI being lower bounded as
\begin{align}
\label{cap4}
\beta \geq D+1.
\end{align}

Blasiak \textit{et al.} \cite{ICVLP} found the capacity of SUICP-SNI with $U=D=1$ by using linear programming bounds to be 
\begin{align}
\label{cap2}
C=\frac{\left\lfloor \frac{K}{2}\right\rfloor}{K}. 
\end{align}

In \cite{VaR4}, we showed that the capacity of SUICP-SNI for arbitrary $K$ and $D$ with $U=\text{gcd}(K,D+1)-1$ is 
\begin{align}
\label{cap3}
C=\frac{1}{D+1}.
\end{align}

\subsection{Review of the known upperbounds} 
In this subsection, we present the known bounds on the broadcast rate of the index coding problems. Let $G$ be the side-information graph of the index coding problem. Let $V(G)$ be the vertex set of the graph $G$.
\subsection*{Broadcast rate and fractional clique cover $\overline{\omega}_f(G)$}

Clique number $\omega(G)$ of a graph $G$ is the size of the largest possible complete subgraph in $G$. Clique cover number $\overline{\omega}(G)$ of a graph G is the minimum number of cliques (complete subgraphs) required to cover the complete vertex set of G. An independent set is the set of vertices of the graph such that no two vertices in this set are adjacent. A fractional clique of a graph G is a non negative real valued function on $V(G)$ such that the sum of the values of the function on the vertices of any independent set is at most one. Fractional clique number $\omega_{f}(G)$ of an undirected graph $G$ is the maximum possible weight of a fractional clique. Blasiak \textit{et al.} in \cite{ICVLP} defined fractional clique cover as a function that assigns a non negative weight to each clique such that for every vertex $x_k \in V(G)$ the total weight assigned to clique containing $x_k$ is atleast one. Fractional clique cover number $\overline{\omega}_{f}(G)$ is defined to be the sum of the weights assigned to the cliques.  

%
%

Blasiak \textit{at~al} in \cite{ICVLP} proved the following upper-bound on broadcast rate
\begin{align}
\label{fcc}
\beta(G) \leq \overline{\omega}_f(G).
\end{align}
\subsection*{Broadcast rate and partial clique cover}
In \cite{ISCO}, Birk and Kol defined partial clique and gave a coding scheme for a given index coding problem based on the partial cliques of the side-information graph. A directed graph $G(V,E)$ is a $k$-partial clique $Clq(s,k)$ iff $\vert V \vert=s$, outdeg$(v) \geq (s-1-k)$, $\forall \ v \in V$, and there exists a $\ v \in V$ such that outdeg$(v)=(s-1-k)$. 


Tehrani \textit{at~al} in \cite{bipartiate} studied bipartite index coding by generalizing partial clique cover scheme. The bipartite index coding is a message partitioning scheme where the sum of the minimum out-degrees are maximized. Assume that, we partition the graph into $l$ disjoint subgraphs induced by message sets $\mathcal{A}_1,\mathcal{A}_2,\ldots,\mathcal{A}_l$, each with the minimum knowledge (out-degree) $d_i$ for $i=1,2,\ldots,l$. Then the subgraph $G_{\mathcal{A}_i}$ can be resolved in $\vert \mathcal{A}_i\vert-d_i$ transmissions and all receivers obtain their wanted messages in $\sum_{i=1}^l (\vert \mathcal{A}_i\vert-d_i)=K-\sum_{i=1}^l d_i$ transmissions. Thus, the optimal partitioning is the solution of the following optimization problem.
\begin{align}
\label{bic}
\nonumber
&\text{maximize} \sum_{i=1}^l d_i \\&
\nonumber
\text{subject ~to} 1 \leq l \leq K ~~\text{and}
\\&\mathcal{A}_1,\mathcal{A}_2,\ldots,\mathcal{A}_l ~\text{is~a~valid~message~decomposition}
\end{align}

Tehrani \textit{at~al} used maximum distance separable (MDS) codes to prove that the proposed upperbound is achievable.
\subsection*{Broadcast rate and fractional local chromatic number $\chi_{l}(\overline{G})$}

The local chromatic number of a directed graph was defined by Korner \textit{et al.} in \cite{korner}. Shanmugam \textit{et al.} in \cite{localcn} used local chromatic number to derive new upper-bound for the broadcast rate of an index coding problem. Let $\overline{G}$ be the complement graph of the graph $G$. 

For a directed graph $\overline{G}$, a coloring of vertices is proper if for every vertex $x_k$ of $\overline{G}$, the color of any of its out-neighbors is different from the color of vertex $x_k$. 
Let $N^+(k)$ be closed outer-neighborhood of a given vertex $x_k$ in the directed graph $\overline{G}$, i.e. $x_j \in N^+(x_k)$ if $(x_k,x_j)$ is a directed edge or $x_k=x_j$. The local chromatic number of a directed graph $\overline{G}$ is the maximum number of colors in any out-neighborhood minimized over all proper colorings of the undirected graph $\overline{G}_u$ obtained from $\overline{G}$ by ignoring the orientation of edges in $\overline{G}$. Let $c: V \rightarrow \{1,2,\ldots,s\}$ be any proper coloring for $\overline{G}_u$ for some integer $s$. Let $|c(N^+(x_k))|$ be the number of colors in the closed out neighborhood of the graph $\overline{G}$. Then,
\begin{align}
\label{lcno}
\chi_l(\overline{G})=\min_{c} \max_{x_k \in V(\overline{G})} |c(N^+(x_k))|.   
\end{align}

Shanmugam \textit{et al.} in \cite{localcn} defined the fractional local chromatic number $\chi_{f_l}(\overline{G})$ of a directed graph $\overline{G}$ and proved that the broadcast rate $\beta(G)$ of an index coding problem defined by side-information graph $G$ is upperbounded by the fractional local chromatic number of $\overline{G}$, that is 
\begin{align}
\label{flcn}
\beta(G) \leq \chi_{f_l}(\overline{G})\leq \chi_{l}(\overline{G}).
\end{align}

Shanmugam \textit{et al.} used MDS codes to prove that the proposed upperbound is achievable.


\subsection*{Broadcast rate and fractional local partial clique covering}

Arbabjolfaei \textit{et al.} in \cite{localtime}, combined the ideas of partial clique covering and fractional local clique covering to establish the fractional local partial clique covering bound.

All upperbounds presented above are graph theory based, whereas, the upper bound presented in this paper is derived by using the properties of extended Euclid algorithm. 

The upperbounds proposed in \cite{bipartiate},\cite{localcn} and \cite{localtime} use MDS codes to prove that the upperbound is achievable. The upperbound presented in this paper can be achieved over every finite field.
\subsection{Review of AIR matrices}
In \cite{VaR2}, we constructed binary matrices of size $m \times n (m\geq n)$ such that any $n$ adjacent rows of the matrix are linearly independent over every finite field. We refer these matrices as AIR matrices.

The matrix obtained by Algorithm \ref{algo2} is called the $(m,n)$ AIR matrix and it is denoted by $\mathbf{L}_{m\times n}.$ The general form of the $(m,n)$ AIR matrix is shown in   Fig. \ref{fig1}. It consists of several submatrices (rectangular boxes) of different sizes as shown in Fig.\ref{fig1}. 
\begin{figure*}
\centering
\includegraphics[scale=0.53]{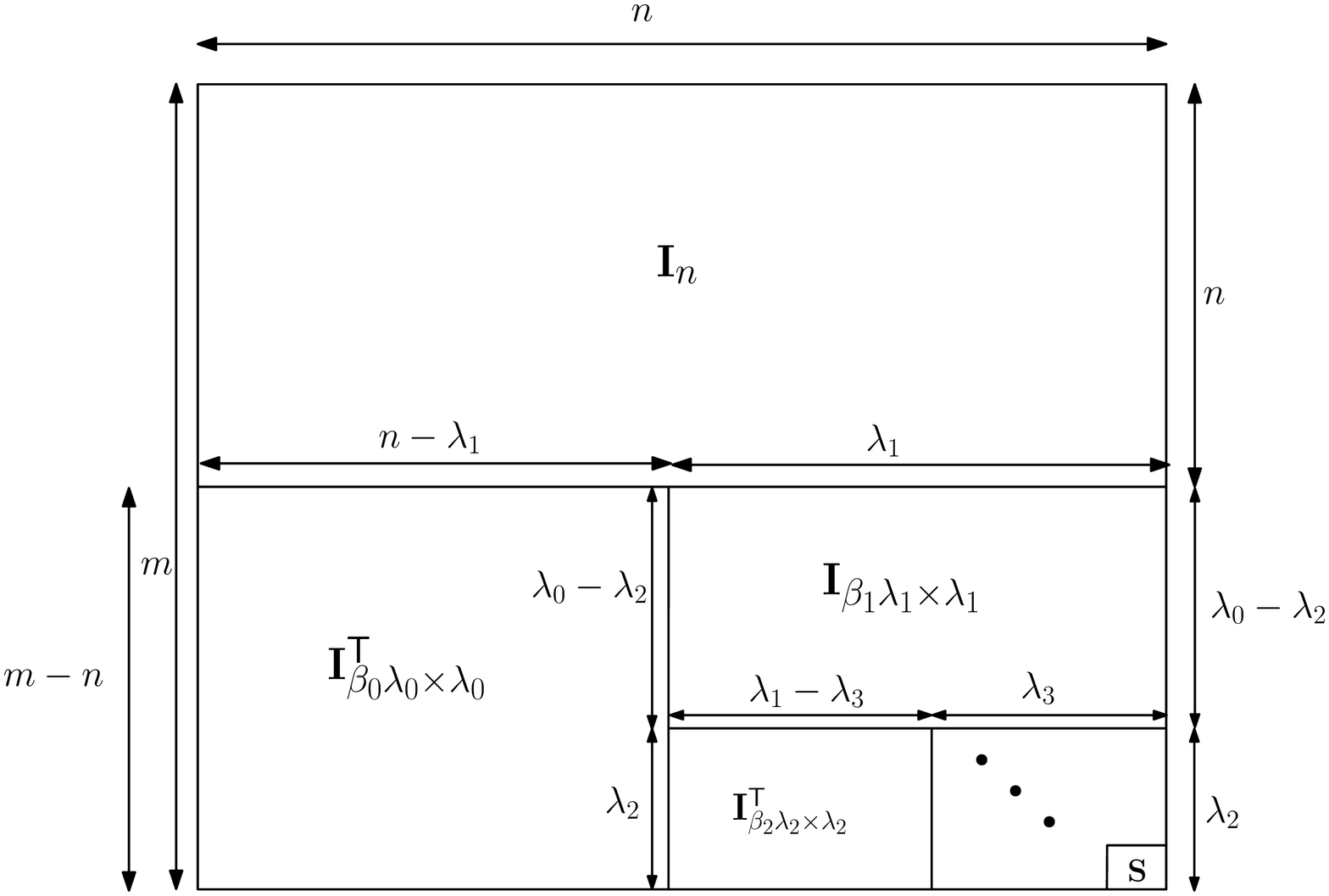}\\
~ $\mathbf{S}=\mathbf{I}_{\lambda_{l} \times \beta_l \lambda_{l}}$ if $l$ is even and ~$\mathbf{S}=\mathbf{I}_{\beta_l\lambda_{l} \times \lambda_{l}}$ otherwise.
\caption{AIR matrix of size $m \times n$.}
\label{fig1}
~ \\
\hrule
\end{figure*}
The description of the submatrices are as follows: Let $c$ and $d$ be two positive integers and $d$ divides $c$. The following matrix  denoted by $\mathbf{I}_{c \times d}$ is a rectangular  matrix.
\begin{align}
\label{rcmatrix}
\mathbf{I}_{c \times d}=\left.\left[\begin{array}{*{20}c}
   \mathbf{I}_{d}  \\
   \mathbf{I}_{d}  \\
   \vdots  \\
   \mathbf{I}_{d} 
   \end{array}\right]\right\rbrace \frac{c}{d}~\text{number~of}~ \mathbf{I}_{d}~\text{matrices}
\end{align}
and $\mathbf{I}_{d \times c}$ is the transpose of $\mathbf{I}_{c \times d}.$

Towards explaining the other quantities shown in the AIR matrix shown in Fig. \ref{fig1}, for a given $m$  and $n,$ let  $\lambda_{-1}=n,\lambda_0=m-n$ and\begin{align}
\nonumber
n&=\beta_0 \lambda_0+\lambda_1, \nonumber \\
\lambda_0&=\beta_1\lambda_1+\lambda_2, \nonumber \\
\lambda_1&=\beta_2\lambda_2+\lambda_3, \nonumber \\
\lambda_2&=\beta_3\lambda_3+\lambda_4, \nonumber \\
&~~~~~~\vdots \nonumber \\
\lambda_i&=\beta_{i+1}\lambda_{i+1}+\lambda_{i+2}, \nonumber \\ 
&~~~~~~\vdots \nonumber \\ 
\lambda_{l-1}&=\beta_l\lambda_l.
\label{chain}
\end{align}
where $\lambda_{l+1}=0$ for some integer $l,$ $\lambda_i,\beta_i$ are positive integers and $\lambda_i < \lambda_{i-1}$ for $i=1,2,\ldots,l$. The number of submatrices in the AIR matrix is $l+2$ and the size of each submatrix is shown using $\lambda_i,\beta_i,$  $i \in [0:l].$

		\begin{algorithm}
		\caption{Algorithm to construct the AIR matrix $\mathbf{L}$ of size $m \times n$}
			\begin{algorithmic}[2]
				 \item Let $\mathbf{L}=m \times n$ blank unfilled matrix.
				\item [Step 1]~~~
				\begin{itemize}
				\item[\footnotesize{1.1:}] Let $m=qn+r$ for $r < n$.
				\item[\footnotesize{1.2:}] Use $\mathbf{I}_{qn \times n}$ to fill the first $qn$ rows of the unfilled part of $\mathbf{L}$.
				\item[\footnotesize{1.3:}] If $r=0$,  Go to Step 3.
				\end{itemize}

				\item [Step 2]~~~
				\begin{itemize}
				\item[\footnotesize{2.1:}] Let $n=q^{\prime}r+r^{\prime}$ for $r^{\prime} < r$.
				\item[\footnotesize{2.2:}] Use $\mathbf{I}_{q^{\prime}r \times r}^{\mathsf{T}}$ to fill the first $q^{\prime}r$ columns of the unfilled part of $\mathbf{L}$.
			    \item[\footnotesize{2.3:}] If $r^{\prime}=0$,  go to Step 3.	
				\item[\footnotesize{2.4:}] $m\leftarrow r$ and $n\leftarrow r^{\prime}$.
				\item[\footnotesize{2.5:}] Go to Step 1.
				\end{itemize}
				\item [Step 3] Exit.
		
			\end{algorithmic}
			\label{algo2}
		\end{algorithm}
			
In \cite{VaR2}, we gave an optimal scalar linear index code for the single unicast index coding problems with symmetric neighboring consecutive side information (SUICP-SNC)(one-sided) using AIR encoding matrices. In \cite{VaR1}, we constructed optimal vector linear index codes for  SUICP-SNC (two-sided). In \cite{VaR3}, we gave a low-complexity decoding for SUICP-SNC with AIR matrix as encoding matrix. The low complexity decoding method helps to identify a reduced set of side-information for each user with which the decoding can be carried out. By this method every receiver is able to decode its wanted message symbol by simply adding some index code symbols (broadcast symbols).

\subsection{Contributions}
Jafar \cite{TIM} established the relation between index coding problem and topological interference management problem. The capacity and optimal coding results in index coding can be used in corresponding topological interference management problems. In \cite{VaR5}, for SUICP-SNI with arbitrary $K,D$ and $U$, we define a set $\mathcal{\mathbf{S}}$ of $2$-tuples such that for every $(a,b) \in \mathcal{\mathbf{S}}$, the rate $D+1+\frac{a}{b}$ is achievable by using AIR matrices with vector linear index codes over every finite field. The contributions of this paper are summarized below:
\begin{itemize}
\item We give an algorithm to find the values of $a$ and $b$ such that $(a,b) \in \mathcal{\mathbf{S}}$ and $\frac{a}{b}$ is minimum.

\item We give an upperbound $R_{airm}(K,D,U)$ on the broadcast rate of SUICP-SNI. We prove that $R_{airm}(K,D,U)$ coincide with the existing results on the exact value of the capacity of SUICP-SNI given in \eqref{cap1},\eqref{cap2} and \eqref{cap3} in the respective settings. 
\end{itemize}


Henceforth,  we refer SUICP-SNI with $K$ messages, $D$ interfering messages after and $U$ interfering messages before the desired message as $(K,D,U)$ SUICP-SNI.

The remaining part of this paper is organized as follows. In Section \ref{sec2}, for $(K,D,U)$ SUICP-SNI, we define a set $\mathbf{S}$ of $2$-tuples such that for every $(a,b) \in \mathcal{\mathbf{S}}$, the rate $D+1+\frac{a}{b}$ is achievable by using AIR matrices with vector linear index codes over every finite field (Theorem \ref{thm1}). In Section \ref{sec3}, we give an algorithm to find the values of $a$ and $b$ such that $(a,b) \in \mathcal{\mathbf{S}}$ and $\frac{a}{b}$ is minimum. In  Section \ref{sec4}, we prove some properties of $R_{airm}(K,D,U)$.
We conclude the paper in Section \ref{sec5}.

All the subscripts in this paper are to be considered $~\text{\textit{modulo}}~ K$. 
\section{Vector Linear Index Codes of SUICP-SNI: Achievability Results}
\label{sec2}
In this section, we define a set $\mathcal{\mathbf{S}}$ consisting of pairs of integers $(a,b)$ and prove that the rate $D+1+\frac{a}{b}$ for every $(a,b) \in \mathcal{\mathbf{S}}$ is achievable by using an appropriate sized AIR matrix as the encoding matrix.
\begin{definition}
\label{def1}
Consider the SUICP-SNI with $K$ messages, $D$ and $U$ interfering messages after and before the desired message. For this SUICP-SNI, define the set $\mathbf{S}_{K,D,U}$ as  
\begin{align}
\label{ab}
\mathbf{S}_{K,D,U}=\{(a,b):\text{gcd}(bK,b(D+1)+a)\geq b(U+1)\}
\end{align}
for $a \in Z_{\geq 0}$ and $b \in Z_{>0}$. 
\end{definition}


\begin{theorem}
\label{thm1}
Consider a $(K,D,U)$ SUICP-SNI. For this index coding problem, for every $(a,b) \in \mathbf{S}_{K,D,U}$, the rate $D+1+\frac{a}{b}$ can be achieved by $b$-dimensional vector linear index coding by using the AIR matrix of size $Kb \times (b(D+1)+a)$. 
\end{theorem}
\begin{proof}
Proof is given in \cite{VaR5}
\end{proof}
\begin{remark}
The AIR encoding matrix $\mathbf{L}_{Kb \times (b(D+1)+a)}$ is an encoding matrix over every finite field. Hence, the encoding for $(K,D,U)$ SUICP-SNI given in Theorem \ref{thm1} is independent of field size.
\end{remark}
\section{Algorithm to find $R_{airm}(K,D,U)$ }
\label{sec3}

\begin{definition}
\label{def2}
Consider the $(K,D,U)$ SUICP-SNI. For this SUICP-SNI, define the sets $\mathbf{S},\mathbf{S}_{r},\mathbf{S}^\prime$ and $\mathbf{S}_r^\prime$ as  
\begin{align}
\label{ab1}
\mathbf{S}=\{(a,b):\text{gcd}(bK,b(D+1)+a)\geq b(U+1)\}
\end{align}
for $a,b \in Z_{>0}$, 
\begin{align}
\label{ab2}
\mathbf{S}_{r}=\{\frac{a}{b}:(a,b) \in \mathbf{S}\},
\end{align}
\begin{align}
\label{ab3}
\nonumber
\mathbf{S}^\prime=\{(a,b):&\text{gcd}(bK,b(D+1)+a)=\\&~~~~~~~~~\text{gcd}(b,m^\prime)K\geq b(U+1)\}
\end{align}
for $m^\prime \in Z_{> 0}$ such that  $b(D+1)+a=m^\prime K$ and 
\begin{align}
\label{ab4}
\mathbf{S}_r^\prime=\{\frac{a}{b}:(a,b) \in \mathbf{S}^\prime\}.
\end{align}
\end{definition}

\begin{lemma}
Let $\mathbf{S}_{r}$ and $\mathbf{S}_{r}^\prime$ be the sets defined in \eqref{ab2} and \eqref{ab4} respectively. Then,
\begin{align*}
\mathbf{S}_{r}= \mathbf{S}_r^\prime.
\end{align*}
\end{lemma}
\begin{proof}
If $\frac{a}{b} \in \mathbf{S}_r$, we prove that $\frac{a}{b} \in \mathbf{S}^\prime_r$. To prove this, we prove that for every $(a,b) \in \mathbf{S}$, there exists $(K a,K b) \in \mathbf{S}^\prime$.

Let $(a,b) \in \mathbf{S}$. Let 
\begin{align}
\label{set4}
\text{gcd}(bK,b(D+1)+a)=b(U+1)+c
\end{align}
for some $c \in Z_{\geq 0}$.

From \eqref{set4}, by multiplying both the sides with $K$, we have 
\begin{align*}
&\text{gcd}( (Kb)K,(Kb)(D+1)+(Ka))=\\&\text{gcd}(b^\prime K,b^\prime(D+1)+a^\prime)=K\text{gcd}(b^\prime,m^\prime)=b^\prime(U+1)+c^\prime,
\end{align*}
where $a^\prime=Ka,b^\prime=Kb$ and $c^\prime=Kc$.
Hence, $(a^\prime, b^\prime) \in \mathbf{S}^\prime$ and $\frac{a^\prime}{b^\prime}=\frac{a}{b} \in \mathbf{S}^\prime_{r}$. We have
\begin{align*}
\mathbf{S}_{r} \subseteq \mathbf{S}^\prime_{r}.
\end{align*}

If $(c,d) \in  \mathbf{S}_r^\prime$, this $(c,d)$ also satisfy the condition $\text{gcd}(dK,d(D+1)+c) \geq d(U+1)$. We have $(c,d) \in \mathbf{S}$ and $\frac{c}{d} \in  \mathbf{S}^\prime$. Hence, we have
\begin{align*}
\mathbf{S}^\prime_{r} \subseteq \mathbf{S}_{r}.
\end{align*}

This completes the proof.
\end{proof}

\begin{lemma}
\label{lemmagcd}
Let $\frac{\alpha_1}{\alpha_2} \in  \mathbf{S}_r^\prime$ and $\text{gcd}(\alpha_1,\alpha_2)=1$. Let $\alpha_1=\gamma K-\alpha_2 (D+1)$ for some $\gamma \in Z_{>0}$. Then, $\text{gcd}(\gamma,\alpha_2)=1$.
\end{lemma}
\begin{proof}
Let $\text{gcd}(\gamma,\alpha_2)>1$. We have
\begin{align}
\label{set25}
\nonumber
\alpha_1&=\text{gcd}(\gamma,\alpha_2)\frac{\gamma}{\text{gcd}(\gamma,\alpha_2)}K-\\& \text{gcd}(\gamma,\alpha_2)\frac{\alpha_2}{\text{gcd}(\gamma,\alpha_2)}(D+1).
\end{align}
From \eqref{set25}, $\text{gcd}(\gamma,\alpha_2)$ is factor of $\alpha_1$ and $\text{gcd}(\alpha_1,\alpha_2)\geq \text{gcd}(\gamma,\alpha_2) > 1$, which is a contradiction. Hence, $\text{gcd}(\gamma,\alpha_2)=1$.
\end{proof}

For the given positive integers $K$ and $D$, extended Euclidean algorithm can be used to find the coefficients of Bezout's identity $m$ and $n$ such that 
\begin{align}
\label{gcd1}
\text{gcd}(K,D+1)=mK-n(D+1).
\end{align}

The integers $K$ and $D$ can be written as
\begin{align}
\label{gcd2}
0=\frac{(D+1)}{\text{gcd}(K,D+1)}K-\frac{K}{\text{gcd}(K,D+1)}(D+1).
\end{align}
The equation 
\begin{align}
\label{gcd3}
x=m^{\prime}K-n^{\prime}(D+1)
\end{align}
has no integer solution \cite{nt} if $x$ is not a integer multiple of $\text{gcd}(K,D+1)$ and has infinite number of solutions if $x$ is a integer multiple of $\text{gcd}(K,D+1)$. If $x=l\text{gcd}(K,D+1)$ for any $l \in Z_{\geq 0}$, then the infinitely many solutions to \eqref{gcd3} can be found by adding $l$ times \eqref{gcd1} with $t$ times \eqref{gcd2} for any $t \in Z$ and the solutions are 
\begin{align}
\label{gcd4}
\nonumber
&m^{\prime}=lm + t\frac{D+1}{\text{gcd}(K,D+1)}~~ \text{and} \\& n^{\prime}=ln + t\frac{K}{\text{gcd}(K,D+1)}.
\end{align}

From \eqref{gcd3}, we have 
\begin{align}
\label{gcd41}
\nonumber
&l\text{gcd}(K,D+1)=m^{\prime}K-n^{\prime}(D+1)~\text{and}~\\&
m^{\prime}K=n^{\prime}(D+1)+l\text{gcd}(K,D+1).
\end{align}
\begin{definition}
\label{def5}
Let $a_{min}=\min_{(a,b) \in \mathbf{S}^\prime}~a.$
\end{definition}
\begin{lemma}
There exists only one $b$ such that $(a_{min},b) \in \mathbf{S}^\prime$.
\end{lemma}
\begin{proof}
Let $(a_{min},b),(a_{min},b^\prime) \in \mathbf{S}^\prime$. Without loss of generality, we assume $b^\prime>b$. Let 
\begin{align}
\label{set30}
b(D+1)+a_{min}=mK
\end{align}
and
\begin{align}
\label{set31}
b^\prime(D+1)+a_{min}=m^\prime K
\end{align}
for some $m,m^\prime \in Z_{>0}$. We have
\begin{align}
\label{set16}
\nonumber
\text{gcd}(bK,b(D+1)+a_{min})&=\text{gcd}(b,m)K\\&=b(U+1)+c
\end{align}
for some $c \in Z_{>0}$, and
\begin{align}
\label{set17}
\nonumber
\text{gcd}(b^\prime K,b^\prime(D+1)+a_{min})&=\text{gcd}(b^\prime,m^\prime)K\\&=b^\prime(U+1)+c^\prime
\end{align}
for some $c^\prime \in Z_{>0}$.

From Definition \ref{def5} and Lemma \ref{lemmagcd}, we have
\begin{align}
\label{set35}
\text{gcd}(b,m)=1~\text{and}~\text{gcd}(b^\prime,m^\prime)=1.
\end{align}

From \eqref{set30} and \eqref{set31}, we have
\begin{align*}
(m^\prime-m)K=(b^\prime-b)(D+1).
\end{align*}
From \eqref{gcd4}, we have
\begin{align}
\label{set20}
b^\prime-b=t \frac{K}{\text{gcd}(K,D+1)}
\end{align}
for any $t \in Z_{> 0}$ (if $t=0$, then $b=b^\prime$, we assumed $b^\prime > b$, hence $t>0$). 
We have $U+1 > \text{gcd}(K,D+1)$ (if $U+1\leq  \text{gcd}(K,D+1)$, we can find the scalar linear index code with $b=1,a=0$). From \eqref{set20}, we have 
\begin{align}
\label{set21}
(b^\prime-b)(U+1)>tK.
\end{align}
From \eqref{set16},\eqref{set17} and \eqref{set35}, we have $K=b(U+1)+c=b^\prime(U+1)+c^\prime$ and $(b^\prime-b)(U+1) < K$. This is a contradiction from \eqref{set21}. Hence, there exists only one $b$ such that $(a_{min},b) \in \mathbf{S}^\prime$.
\end{proof}

\begin{definition}
\label{def6}
Define $b_{min}$ as the corresponding value of $a_{min}$ such that $(a_{min},b_{min}) \in \mathbf{S}^\prime$. 
\end{definition}

\begin{theorem}
\label{lemma4}
Let $(a_{min},b_{min}) \in \mathbf{S}^\prime$. Then
\begin{align*}
\frac{a_{min}}{b_{min}}<\frac{a^\prime}{b^\prime}~~\forall \frac{a^\prime}{b^\prime} \in \mathbf{S}_r.
\end{align*}
\end{theorem}
\begin{proof}
Let $\frac{a^\prime}{b^\prime} \in  \mathbf{S}_r$ such that 
\begin{align}
\label{set13}
\frac{a_{min}}{b_{min}} > \frac{a^\prime}{b^\prime}.
\end{align}

With out loss of generality, we assume $\text{gcd}(a^\prime,b^\prime)=1$. From Definition \ref{def5}, we have $a^\prime>a_{min}$. Let
\begin{align}
\label{gcd5}
a_{min}=m K-b_{min}(D+1)
\end{align}
and 
\begin{align}
\label{gcd6}
a^\prime=m^\prime K-b^\prime(D+1)
\end{align}
for some $m,m^\prime \in Z_{>0}$. 

From Lemma \ref{lemmagcd}, we have
\begin{align}
\label{set27}
\text{gcd}(b_{min},m)=1~\text{and}~\text{gcd}(b^\prime,m^\prime)=1.
\end{align}

From \eqref{gcd5},\eqref{gcd6},\eqref{set27} and Definition \ref{def2}, we have 
\begin{align}
\label{set211}
\text{gcd}(b_{min}K,\underbrace{b_{min}(D+1)+a_{min}}_{m K})=K\geq b_{min}(U+1)
\end{align}
and 
\begin{align}
\label{set22}
\text{gcd}(b^\prime K, \underbrace{b^\prime(D+1)+a^\prime}_{m^\prime K})=K \geq b^\prime(U+1).
\end{align}

By subtracting $\left\lfloor\frac{a^\prime}{a_{min}}\right\rfloor$ times of \eqref{gcd5} from \eqref{gcd6}, we have
\begin{align}
\label{set12}
\nonumber
&a^\prime~\text{mod}~a_{min}=\\&(m^\prime-\left\lfloor \frac{a^\prime}{a_{min}}\right\rfloor m)K-(b^\prime-\left\lfloor \frac{a^\prime}{a_{min}}\right\rfloor b_{min})(D+1).
\end{align}

Let 
\begin{align}
\label{set29}
\nonumber
&x=m^\prime-\left\lfloor \frac{a^\prime}{a_{min}}\right\rfloor m~\text{and}\\&y=b^\prime-\left\lfloor \frac{a^\prime}{a_{min}}\right\rfloor b_{min}.
\end{align}
\medskip
\textit{Case (i):}$x < 0$ and $y <0$. 

In this case, we have $y=b^\prime-\left\lfloor \frac{a^\prime}{a_{min}}\right\rfloor b_{min} <0$, hence $\frac{b^\prime}{b_{min}} <\left\lfloor \frac{a^\prime}{a_{min}}\right\rfloor$. This is a contradiction from \eqref{set13}.

\textit{Case (ii):}$x < 0$ and $y >0$. 

In this case, the LHS of \eqref{set12} is positive and RHS is negative, which is a contradiction.

\textit{Case (iii):}$x > 0$ and $y <0$. 

In this case, the LHS of \eqref{set12} is less than $K$ and RHS is greater than $K$, which is a contradiction.

\textit{Case (iv):} $x > 0$ and $y >0$. 

From \eqref{set211} and \eqref{set22}, we have 
\begin{align}
\label{set28}
b_{min} \leq \frac{K}{U+1}~\text{and}~b^\prime \leq\frac{K}{U+1}.
\end{align}

In this case, from \eqref{set29} and \eqref{set28}, we have $0<y \leq \frac{K}{U+1}$. From \eqref{set12}, we have
\begin{align}
\text{gcd}(y K,\underbrace{y(D+1)+a^\prime~\text{mod}~a_{min}}_{xK})&=\text{gcd}(x,y)K\\& \geq y (U+1).
\end{align}

Hence, $(a^\prime~\text{mod}~a_{min},y) \in \mathbf{S}^\prime$. This is a contradiction from the definition of $a_{min}$ because $a^\prime~\text{mod}~a_{min} < a_{min}$.

This completes the proof.
\end{proof}

\begin{definition}
Consider an SUICP-SNI with $K$ messages, $D$ interfering messages after and $U$ interfering messages before. Define $R_{airm}(K,D,U)$ as
\begin{align*}
R_{airm}(K,D,U)=D+1+\frac{a_{min}}{b_{min}}.
\end{align*}
\end{definition}
\begin{theorem}
Consider an SUICP-SNI with $K$ messages, $D$ interfering messages after and $U$ interfering messages before. Let $G$ be the side-information graph of this index coding problem. Then, 
\begin{align*}
\beta(G) \leq R_{airm}(K,D,U).
\end{align*}
\end{theorem}
\begin{proof}
From the definition of $a_{min}$ and $b_{min}$, the tuple $(a_{min},b_{min}) \in \mathcal{\mathbf{S}}$. From Theorem \ref{thm1}, the rate $D+1+\frac{a_{min}}{b_{min}}=R_{airm}(K,D,U)$ can be achieved by using AIR matrices. Hence, we have $\beta(G) \leq R_{airm}(K,D,U).$
\end{proof}

Algorithm \ref{algo1} computes the values of $a_{min}$ and $b_{min}$ for a given $(K,D,U)$ SUICP-SNI. For the given values of $K,D$ and $U$,  Algorithm \ref{algo1} computes $m^{\prime}$ and $n^{\prime}$ given in \eqref{gcd4} for $l=1$ and $t \in [-l:l]$. If $n^{\prime} \in \{1,2,\ldots,\left\lfloor\frac{K}{U+1}\right\rfloor\}$ for any $t \in [-1:1]=\{-1,0,1\}$, the algorithm terminates and outputs $a_{min}=\text{gcd}(K,D+1)$ and $b_{min}=n^\prime$, else the Algorithm increases the value of $l$ by one and repeats Step2 until it finds an $n^{\prime} \in \{1,2,\ldots,\left\lfloor\frac{K}{U+1}\right\rfloor\}$. To compute the values of $a_{min}$ and $b_{min}$, in Lemma \ref{lemma6}, we prove that Algorithm \ref{algo1} terminates for some $l \leq \frac{K \text{mod} (D+1)}{\text{gcd}(K,D+1)}$.

		\begin{algorithm*}[ht]
		\caption{Algorithm to find $a_{min}$ and $b_{min}$}
			\begin{algorithmic}[1]
				\item [Step 1]~~~
				\begin{itemize}
				\item[\footnotesize{1.1:}] $l=1$.
				\item[\footnotesize{1.2:}] $m$ and $n$ are the coefficients of Bezout's identity such that $mK-n(D+1)=\text{gcd}(K,D+1)$.
				\end{itemize}
				\item [Step 2]~~~
				\begin{itemize}
				\item[\footnotesize{2.1:}] $t=-l$

				\end{itemize}
				\item [Step 3]~~~
				\begin{itemize}
				 \item[\footnotesize{3.1:}] $m^{\prime}=lm + t\frac{D+1}{\text{gcd}(K,D+1)}$ and  $n^{\prime}=ln + t\frac{K}{\text{gcd}(K,D+1)}$.
				\item[\footnotesize{3.2:}] \textbf{If} {$n^{\prime} \in [1:\left \lfloor \frac{K}{U+1} \right \rfloor]$}, then
				$a_{min}=lgcd(K,D+1)$ and $b_{min}=n^{\prime}$ 
				\item[\footnotesize{3.3:}] exit.
			    \item[\footnotesize{3.4:}] \textbf{If} $t < l$, then $t=t+1$ and Repeat Step 3.
			    	\item[\footnotesize{3.5:}] \textbf{else} $l=l+1$. 
			    \item[\footnotesize{3.5:}]  Repeat Step 2.
		     \end{itemize}
		
			\end{algorithmic}
			\label{algo1}
		\end{algorithm*}
 
\begin{lemma}
\label{lemma6}
In Algorithm \ref{algo1},
\begin{align*}
l \leq \frac{K \text{mod} (D+1)}{\text{gcd}(K,D+1)}. 
\end{align*}
\end{lemma}
\begin{proof}
We have
\begin{align}
\label{gcd8}
\frac{K \text{mod}(D+1)}{\text{gcd}(K,D+1)}\text{
gcd}(K,D+1)=K-\left \lfloor \frac{K}{D+1}\right \rfloor(D+1).
\end{align}
From \eqref{gcd8}, we have $\left \lfloor \frac{K}{D+1}\right \rfloor \in [1:\left \lfloor \frac{K}{U+1}\right \rfloor]$ and
\begin{align*}
(K ~\text{mod}~(D+1), \left \lfloor \frac{K}{D+1}\right \rfloor) \in \mathbf{S}^\prime.
\end{align*} 
Hence, $a_{min} \leq K \text{mod}(D+1)$ and $l \leq \frac{K \text{mod}(D+1)}{\text{gcd}(K,D+1)}$. This completes the proof.
\end{proof}

\begin{example}
Consider a SUICP-SNI with $K=17,D=11,U=1$. We have $\text{gcd}(K,D+1)=\text{gcd}(17,12)=1<U+1=2$. From the Extended Euclidean algorithm, the coefficients of Bezout's identity are 5 and -7. We have 
\begin{align*}
&1=5\times 17-7\times 12 \\&
0=12 \times 17-17 \times 12.
\end{align*}
\begin{itemize}
\item Let $l=1$. We have $n^{\prime}=-12,7,24,41,58,\ldots$ and $n^{\prime}=7 \in [1:\left \lfloor \frac{K}{U+1} \right \rfloor]=\{1,2,\ldots, 9\}$.
\end{itemize}

Hence, Algorithm \ref{algo1} gives  $a_{min}=l\text{gcd}(K,D+1)=1$ and $b_{min}=n^{\prime}=7$ as output. For this index coding problem, $R_{airm}=D+1+\frac{a_{min}}{b_{min}}=12.142$. The AIR matrix of size $119 \times 85$ can be used as an  encoding matrix for this SUICP-SNI to achieve a rate of $R_{airm}=12.142$.
\end{example}
\begin{example}
Consider a SUICP-SNI with $K=17,D=5,U=1$. We have $\text{gcd}(K,D+1)=\text{gcd}(17,6)=1<U+1=2$. From the extended Euclidean algorithm, the coefficients of Bezout's identity are -1 and 3. We have  
\begin{align*}
&1=-1\times 17+3\times 6 \\&
0=6 \times 17-17 \times 6.
\end{align*}

\begin{itemize}
\item Let $l=1$. We have $n^{\prime}=-3,14,31,48,\ldots$ and $n^{\prime} \notin [1:\left \lfloor \frac{K}{U+1} \right \rfloor]= \{1,2,\ldots, 9\}$.
\item Let $l=2$. We have $n^{\prime}=-6,11,28,45,\ldots$ and $n^{\prime} \notin \{1,2,\ldots, 9\}$.
\item Let $l=3$. We have $n^{\prime}=-9,8,25,42,\ldots$ and $n^{\prime}=8 \in \{1,2,\ldots, 9\}$.
\end{itemize}
Hence, Algorithm \ref{algo1} gives $a_{min}=l\text{gcd}(K,D+1)=3$ and $b_{min}=n^{\prime}=8$ as output. For this index coding problem, $R_{airm}=D+1+\frac{a_{min}}{b_{min}}=6.375$. The AIR matrix of size $136 \times 51$ can be used as an encoding matrix for this SUICP-SNI to achieve a rate of $R_{airm}=6.375$.
\end{example}

\begin{example}
\label{ex2}
Consider a SUICP-SNI with $K=71,D=25,U=1$. For this SUICP-SNI, we have $a_{min}=1,b_{min}=30$ and  corresponding $R_{airm}(71,25,1)=26.033$. This rate can be achieved by the AIR matrix of size $2130 \times 781$. The encoding matrix for this SUICP-SNI is shown below.

\vspace{10pt}
\centering
\includegraphics[scale=0.56]{matrix122.eps}\\
\end{example}
\section{Properties of $R_{airm}(K,D,U)$}
\label{sec4}
In this section, we derive some properties of $R_{airm}(K,D,U)$.
\begin{lemma}
\label{lemma61}
For every  $(K,D,U)$ SUICP-SNI, 
\begin{align}
\label{gcd84}
R_{airm}(K,D,U) \leq \frac{K}{\left \lfloor \frac{K}{D+1} \right \rfloor}.
\end{align}
\end{lemma}
\begin{proof}
We have 
\begin{align}
\label{gcd83}
\nonumber
\frac{K}{\left\lfloor\frac{K}{D+1}\right\rfloor}&=\frac{\left\lfloor\frac{K}{D+1}\right\rfloor(D+1)+K \text{mod} (D+1)}{\left\lfloor\frac{K}{D+1}\right\rfloor} \\& 
=D+1+\frac{K \text{mod} (D+1)}{\left\lfloor\frac{K}{D+1}\right\rfloor}=D+1+\frac{\alpha}{\gamma},
\end{align}
where $\alpha=K \text{mod} (D+1)$ and $\gamma=\left\lfloor\frac{K}{D+1}\right\rfloor$.

From \eqref{gcd83}, we have
\begin{align}
\label{gcd81}
\alpha=K-\gamma(D+1)
\end{align}
and these values of $\alpha$ and $\gamma$ satisfy the equation $\text{gcd}(K \gamma,\gamma(D+1)+\alpha)\geq \gamma (U+1)$. Hence, $(\alpha,\gamma) \in \mathcal{\mathbf{S}}$. From the definition of $a_{min}$ and $b_{min}$, we have $\frac{a_{min}}{b_{min}} \leq \frac{\alpha}{\gamma}$. This completes the proof.
\end{proof}
\begin{corollary}
\label{cor2}
For, SUICP-SNI, the vector linear index codes constructed by AIR matrices are within $\frac{K \text{mod} (D+1)}{\left\lfloor\frac{K}{D+1}\right\rfloor}$ symbols per message from the lower bound on broadcast rate given in \eqref{cap4}.
\end{corollary}
\begin{theorem}
\label{thm2}
The rate $R_{airm}(K,D,U)$ coincide with the results on the exact capacity of SUICP-SNI given in \eqref{cap1},\eqref{cap2} and \eqref{cap3} in the respective settings. 
\end{theorem}
\begin{proof}
\begin{enumerate}
\item []
\item To recover the result corresponding to  \eqref{cap1}:
For a given $U \leq D$, if $K \rightarrow \infty$, then $\frac{K \text{mod} (D+1)}{\left\lfloor\frac{K}{D+1}\right\rfloor} \rightarrow 0$. In Lemma \ref{lemma61}, we proved that 
\begin{align*}
D+1+\frac{a_{min}}{b_{min}} \leq D+1+\frac{K \text{mod} (D+1)}{\left\lfloor\frac{K}{D+1}\right\rfloor}.
\end{align*}
Hence, $D+1+\frac{a_{min}}{b_{min}} \rightarrow D+1$. This completes the proof.
\item To recover the capacity corresponding to \eqref{cap2}: 
 Substituting $D=1$ in \eqref{gcd84} completes the proof of this case.
\item  To recover  \eqref{cap3} as a special case:
If we take $a=0$ and $b=1$, then we get a setting considered in \eqref{cap3}. Hence, the scalar linear code ($b=1$) considered in \eqref{cap3} is a special case of vector linear codes considered in this paper.
\end{enumerate}
\end{proof}

For SUICP-SNI with $K=37$, $U \leq D \leq 8$, the values of $a_{min}$, $b_{min}$ and $R_{airm}$ are shown in Table \ref{table1}. For these index coding problems, $D+1$ is given in the $5$th column of Table \ref{table1} gives a lower bound on broadcast rate. The values of $D+1$ can be compared with $R_{airm}$ in Table \ref{table1}. The rate  $R_{airm}$ given in $6$th column is achieved by using AIR matrices of size given in the $7$th column of Table \ref{table1}. For SUICP-SNI with $K=37$, the capacity is known to the special case $U=D=1$. For $U=D=1$, $R_{airm}=2.0555$ coincide with the reciprocal of capacity given in \eqref{cap2}.

{\small
\begin{table}[ht]
\centering
\setlength\extrarowheight{0.0pt}
\begin{tabular}{|c|c|c|c|c|c|c|}
\hline
$D$&$U$&$a$&$b$&$D+1$&$R_{airm}$&AIR  \\
 & &   &  &  &  &  matrix size \\
\hline
 1 & 1&1&18&2&2.055&$666 \times 37$ \\
\hline
 2 & 1,2&1&12&3&3.083&$444 \times 37$ \\
\hline
 3 & 1,2,3&1&9&4&4.111&$333 \times 37$ \\
\hline
 4 & 1,2,3,4&2&7&5&5.285&$259 \times 37$ \\
\hline
 5 &1,2,3,4,5&1&6&6&6.166&$222 \times 37$ \\
\hline
 6 &1,2,$\ldots$,6&2&5&7&7.400&$185 \times 37$ \\
\hline
 7 &1,2,3&2&9&8&8.222&$333 \times 74$ \\
\hline
 7 & 4,5,6,7&5&4&8&9.250&$148 \times 37$ \\
\hline
 8 &1,2,$\ldots$,8&1&4&9&9.250&$148 \times 37$ \\
\hline
\end{tabular}
\vspace{5pt}
\caption{$R_{airm}$ for $K=37$ and $U \leq D \leq 8$.}
\label{table1}
\end{table}
}

\section{Discussion}
\label{sec5}
In this paper, we gave an upperbound on the broadcast rate of SUICP-SNC and proved that this upperbound coincides with the existing three results of capacity of SUICP-SNI. Obtaining the capacity of $(K,D,U)$ SUICP-SNI is a challenging open problem.

In the AIR matrix of size $m \times n$, any set of $n$ adjacent rows are linearly independent, not only over finite fields, but also over real ($\mathbb{R}$) and complex fields ($\mathbb{C}$). Hence, the application of AIR matrices in wireless TIM problems and wireless non-orthogonal multiple access (NOMA) techniques is an interesting area of research.

\section*{Acknowledgment}
This work was supported partly by the Science and Engineering Research Board (SERB) of Department of Science and Technology (DST), Government of India, through J.C. Bose National Fellowship to B. Sundar Rajan


\begin{thebibliography}{160}
\bibitem{MCJ}
H. Maleki, V. Cadambe, and S. Jafar, ``Index coding – an interference alignment perspective", in IEEE \textit{Trans. Inf. Theory,}, vol. 60, no.9, pp.5402-5432, Sep. 2014.
\bibitem{TIM}
S. A. Jafar, “Topological interference management through index
coding,” IEEE Trans. Inf. Theory, vol. 60, no. 1, pp. 529–568,
Jan. 2014.
\bibitem{ISCO}
Y. Birk and T. Kol, ``Informed-source coding-on-demand (ISCOD) over broadcast channels", in \textit{Proc. IEEE Conf. Comput. Commun.}, San Francisco, CA, 1998, pp. 1257-1264.

\bibitem{YBJK}
Z. Bar-Yossef, Z. Birk, T. S. Jayram and T. Kol, ``Index coding with side-information", in \textit{Proc. 47th Annu. IEEE Symp. Found. Comput. Sci.}, Oct. 2006, pp. 197-206.


\bibitem{OnH}
L Ong and C K Ho, ``Optimal Index Codes for a Class of Multicast Networks with Receiver Side Information'', in \textit{Proc. IEEE ICC}, 2012, pp. 2213-2218.

\bibitem{minrank}
R. Peeters, “Orthogonal representations over finite fields and the chromatic number of graphs,” \textit{Combinatorica}, vol. 16, no. 3, Sept 1996, pp. 417–431.

\bibitem{ICVLP}
A. Blasiak, R. Kleinberg and E. Lubetzky, ``Broadcasting With side-information: Bounding and Approximating the Broadcast Rate" \textit{CoRR}, in IEEE \textit{Trans. Inf. Theory,}, vol. 59, no.9, pp.5811-5823, Sep. 2013.

\bibitem{nt}
David M. Burton, ``Elementray Number Theory,'' Seventh Edition, Mc Graw Hill publisher.
\bibitem{NL}
E. Lubetzky and U. Stav, “Nonlinear index coding outperforming the linear optimum,” IEEE Trans. Inf. Theory, vol. 55, no. 8, pp. 3544–3551, Aug. 2009.

\bibitem{korner}
J. Korner, C. Pilotto, and G. Simonyi, “Local chromatic number and sperner capacity,” \textit{Journal of Combinatorial Theory}, Series B, vol. 95, no. 1, pp. 101–117, 2005.
\bibitem{bipartiate}
A. S. Tehrani, A. G. Dimakis, and M. J. Neely, ``Bipartite index coding,” in \textit{Proc. IEEE ISIT 2012}, pp. 2246–2250.
\bibitem{localcn}
K. Shanmugam, A. G. Dimakis, and M. Langberg, ``Local graph coloring and index coding", in \textit{Proc. IEEE ISIT 2013}, pp 1152–1156.
\bibitem{localtime}
F. Arbabjolfaei and Y. H. Kim, ``Local time sharing for index coding", in \textit{Proc. IEEE ISIT 2014}, pages 286–290.
\bibitem{AGT}
C. Godsil and G. Royle, \textit{Algebraic Graph Theory}, Springer 2001.

\bibitem{VaR1}
M. B. Vaddi and B. S. Rajan, ``Optimal Vector Linear Index Codes for Some Symmetric Multiple Unicast Problems,'' in \textit{Proc. IEEE ISIT}, Barcelona, Spain, July 2016.

\bibitem{VaR2}
M. B. Vaddi and B. S. Rajan, ``Optimal Scalar Linear Index Codes for One-Sided Neighboring Side-Information Problems,'' \textit{In Proc. IEEE GLOBECOM Workshop on Network Coding and Applications}, Washington, USA, December 2016.

\bibitem{VaR3}
M. B. Vaddi and B. S. Rajan, ``Low-Complexity Decoding for Symmetric, Neighboring and Consecutive Side-information Index Coding Problems,'' in arXiv:1705.03192v2 [cs.IT] 16 May 2017.

\bibitem{VaR4}
M. B. Vaddi and B. S. Rajan, ``Capacity of Some Index Coding Problems with Symmetric Neighboring Interference,'' in arXiv:1705.05060v2 [cs.IT] 18 May 2017.

\bibitem{VaR5} M. B. Vaddi and B. S. Rajan, ``Near-Optimal Vector Linear Index Codes For Single Unicast Index Coding Problems with Symmetric Neighboring Interference,'' in arXiv:	1705.10614  [cs.IT] 28 May 2017.

\end{thebibliography}
\end{document}